\documentclass[11pt, reqno]{article}
\usepackage[hmargin=1.0in,vmargin=1.0in]{geometry}
\usepackage{amsthm}
\usepackage[english]{babel}
\usepackage{cite}
\usepackage{color}
\usepackage{amsmath,amsfonts,amssymb,graphicx}
\usepackage{microtype}
\usepackage{mathrsfs,url}
\usepackage{enumerate}
\usepackage{cmathml}

\DeclareMathOperator{\proj}{\ensuremath{\pi}}

\DeclareMathOperator{\fpol}{fPol}

\renewcommand{\vec}[1]{\ensuremath{\mathbf{#1}}}

\newcommand{\inst}{\ensuremath{I}}

\DeclareMathOperator{\rel}{{\bf R}}
\newcommand{\tup}[1]{\ensuremath{\mathbf #1}}

\newcommand{\ignore}[1]{}
\newcommand{\multiset}[1]{\ensuremath{\{\hspace*{-2pt}\{#1\}\hspace*{-2pt}\}}}

\DeclareMathOperator{\wrel}{{\bf \Phi}}
\DeclareMathOperator{\supp}{{\rm supp}}
\DeclareMathOperator{\CSP}{CSP}
\DeclareMathOperator{\VCSP}{VCSP}

\DeclareMathOperator{\feas}{Feas}
\DeclareMathOperator{\opt}{Opt}

\DeclareMathOperator{\rdom}{{Feas}}

\DeclareMathOperator{\pol}{Pol}
\DeclareMathOperator{\ops}{{\mathcal O}}

\DeclareMathOperator*{\E}{\mathop{{\mathbb E}}}
\DeclareMathOperator{\avg}{{\mathrm avg}}

\newcommand{\qq}{\ensuremath{\overline{\mathbb{Q}}}}

\date{}

\newtheorem{theorem}{Theorem}
\newtheorem{lemma}{Lemma} 
\newtheorem*{lemma*}{Lemma} 
\newtheorem*{proposition*}{Proposition} 
\newtheorem*{theorem*}{Theorem} 
\newtheorem{proposition}[theorem]{Proposition}
\newtheorem{corollary}{Corollary}
\newtheorem{definition}{Definition}
\theoremstyle{remark}
\newtheorem{example}{Example}

\newtheorem{problem}{Problem}

\begin{document}
\title{Sherali-Adams relaxations for valued CSPs\thanks{The authors were supported by
a London Mathematical Society Grant. Stanislav \v{Z}ivn\'y was supported
by a Royal Society University Research Fellowship.}}

\author{
Johan Thapper\\
Universit\'e Paris-Est, Marne-la-Vall\'ee, France\\
\texttt{thapper@u-pem.fr}
\and
Stanislav \v{Z}ivn\'{y}\\
University of Oxford, UK\\
\texttt{standa@cs.ox.ac.uk}
}

\maketitle

\begin{abstract}
We consider Sherali-Adams linear programming relaxations for solving
valued constraint satisfaction problems to optimality.
The utility of linear programming relaxations in this context have previously been demonstrated
using the lowest possible level of this hierarchy under the name of the basic
linear programming relaxation (BLP).
It has been shown that valued constraint languages containing only finite-valued weighted relations
are tractable if, and only if, the integrality gap of the BLP is 1.
In this paper, we demonstrate that almost all of the known tractable languages 
with arbitrary weighted relations have an integrality gap 1 for the Sherali-Adams
relaxation with parameters $(2,3)$.
The result is closely connected to the notion of bounded relational width
for the ordinary constraint satisfaction problem
and its recent characterisation.
\end{abstract}

\section{Introduction}

The constraint satisfaction problem  provides a common framework for many
theoretical and practical problems in computer science. An
instance of the \emph{constraint satisfaction problem} (CSP) consists of a
collection of variables that must be assigned labels from a given domain subject
to specified constraints. 
The CSP is NP-complete in general, but tractable fragments can be studied
by, following Feder and Vardi~\cite{Feder98:monotone}, restricting the constraint relations
allowed in the instances to a fixed, finite set, called the constraint language.
The most successful approach to
classifying the language-restricted CSP is the so-called algebraic
approach~\cite{Bulatov05:classifying,Barto14:jacm}.

An important type of algorithms for CSPs are \emph{consistency methods}.
A constraint language is of \emph{bounded relational width} if any CSP
instance over this language can be solved by establishing $(k,\ell)$-minimality for
some fixed integers $1 \leq k\leq\ell$~\cite{Barto14:jloc}. 
The power of consistency methods for
constraint languages has recently been fully
characterised~\cite{LaroseZadori07:au,Barto14:jacm} and it has been shown that
any constraint language that is of bounded relational width is of relational
width at most $(2,3)$\cite{Barto14:jloc}.

The CSP deals with only feasibility issues: Is there a solution satisfying certain
constraints? In this work we are interested in problems that capture both
feasibility and optimisation issues: What is the best solution satisfying
certain constraints? Problems of this form can be cast as valued constraint
satisfaction problems~\cite{jkz14:beatcs}.

An instance of the \emph{valued constraint satisfaction problem} (VCSP) is given
by a collection of variables that is assigned labels from a given domain
with the goal to \emph{minimise} an objective function given by a sum
of weighted relations, each depending on some subset of the
variables~\cite{Cohen06:complexitysoft}. The weighted relations can take on finite
rational values and positive infinity. 
The CSP corresponds to the special case of the VCSP when the codomain of all
weighted relations is $\{0,\infty\}$. 

Like the CSP,
the VCSP is NP-hard in general and thus we are interested in the restrictions
which give rise to tractable classes of problems. 
We restrict the \emph{valued constraint language}; that
is, all weighted relations in a given instance must belong to a fixed set of
weighted 
relations on the domain. The ultimate goal is to understand the computational
complexity of all valued constraint languages, that is, determine which
languages give rise to classes of problems solvable in polynomial time and which
languages give rise to classes of problems that are NP-hard. Languages of the
former type are called \emph{tractable}, and languages of the latter type are
called \emph{intractable}.
The computational complexity of Boolean (on a $2$-element domain)
valued constraint languages~\cite{Cohen06:complexitysoft} and conservative (containing
all $\{0,1\}$-valued unary weighted relations) valued constraint languages~\cite{kz13:jacm} have been
completely classified with respect to exact solvability. 

Every VCSP problem has a natural linear programming (LP)
relaxation, proposed independently by a number of authors, e.g.~\cite{Chekuri04:sidma},
and referred to as the \emph{basic} LP relaxation (BLP) of the VCSP.
It is the first level in the Sheralli-Adams hierarchy~\cite{Sherali1990},
which provides successively tighter LP relaxations of an integer LP. 
The BLP has been considered in the context of CSPs for robust
approximability~\cite{Kun12:itcs,Dalmau13:robust} and constant-factor
approximation~\cite{Ene13:soda,Dalmau15:soda}. 
Higher levels of Sherali-Adams hierarchy have been considered for
(in)approximability of CSPs~\cite{Vega07:soda,Yoshida14:itcs} but we are
not aware of any results related to exact solvability of (valued) CSPs.
Semidefinite programming
relaxations have also been considered in the context of CSPs for
approximability~\cite{Raghavendra08:stoc} and robust
approximability~\cite{Barto12:stoc}.

Consistency methods, and in particular strong 3-consistency
has played an important role as a preprocessing step in establishing
tractability of valued constraint languages. Cohen et al. proved the
tractability of valued constraint languages improved by a symmetric tournament
pair (STP) multimorphism via strong 3-consistency preprocessing, and an
involved reduction to submodular function
minimisation~\cite{Cohen08:Generalising}. 
They also showed that the
tractability of any valued constraint language improved by a tournament pair
multimorphism via a preprocessing using
results on constraint
languages invariant under a 2-semilattice polymorphism, which relies on
$(3,3)$-minimality, and then reducing to the STP case. The only tractable
conservative valued constraint languages are those admitting a pair of
fractional polymorphisms called STP and MJN~\cite{kz13:jacm}; again, the
tractability of such languages is proved via a 3-consistency preprocessing
reducing to the STP case. It is natural to ask whether this nested use of
consistency methods are necessary.

\subsection{Contributions}

In~\cite{tz12:focs,ktz15:sicomp}, the authors showed that the
BLP of the VCSP can be used to solve the problem for
many valued constraint languages. In~\cite{tz13:stoc}, it was then shown that
for VCSPs with weighted relations taking only finite values, the BLP precisely
characterises the tractable (finite-)valued constraint languages; i.e., if BLP fails to solve
any instance of some valued constraint language of this type, then this language
is NP-hard.

In this paper, we show that a higher-level Sherali-Adams linear programming
relaxation~\cite{Sherali1990} suffices to solve most of the previously known
tractable valued constraint languages with arbitrary weighted relations, and in
particular, all known valued constraint languages that involve some optimisation
(and thus do not reduce to constraint languages containing only relations)
except for valued constraint languages of generalised weak tournament pair
type~\cite{Uppman13:icalp}; such languages are known to be
tractable~\cite{Uppman13:icalp} but we do not
know whether they are tractable by our linear programming relaxation.

Our main result, Theorem~\ref{thm:main}, shows that if the support clone of a
valued constraint language $\Gamma$ of finite size contains weak near-unanimity
operations of all but finitely many arities, then $\Gamma$ is
tractable via the Sherali-Adams relaxation with parameters $(2,3)$.
This tractability condition is precisely the bounded relational width condition for constraint languages of
finite size containing all constants~\cite{LaroseZadori07:au,Barto14:jacm},
and our proof fundamentally relies on the results of Barto and Kozik~\cite{Barto14:jacm} and Barto~\cite{Barto14:jloc}.

It is
folklore that the $k$th level of Sherali-Adams hierarchy establishes
$k$-consistency for CSPs. We demonstrate that one linear
programming relaxation is powerful enough to establish consistency as well as
solving an optimisation problem in one go without the need of nested
applications of consistency methods.
For example, valued constraint languages having a tournament pair multimorphism
were previously known to be tractable using ingenious application of various
consistency techniques, advanced analysis of constraint networks using modular
decompositions, and submodular function
minimisation~\cite{Cohen08:Generalising}. Here, we show that an even less
restrictive condition (having a binary conservative commutative operation in
some fractional polymorphism) ensures that the Sherali-Adams relaxation solves
all instances to optimum.

Finally, we also give a short proof of the dichotomy theorem for conservative valued
constraint languages~\cite{kz13:jacm}, which previously needed lengthy arguments
(although we still rely on Takhanov~\cite{Takhanov10:stacs} for a part of the
proof). 

\section{Preliminaries}

\subsection{Valued CSPs}

Throughout the paper, let $D$ be a fixed finite set of size at least two.

\begin{definition}\label{def:rel}
An $m$-ary \emph{relation} over $D$ is any mapping $\phi:D^m\to\{c,\infty\}$ for
some $c\in\mathbb{Q}$. 
We denote by $\rel_D$ the set of all relations on $D$.\footnote{
An $m$-ary relation over $D$ is commonly defined as a subset of $D^m$. Note that
Definition~\ref{def:rel} is equivalent to the standard definition as any mapping
$\phi:D^m\to\{c,\infty\}$ represents the set $R=\{\tup{x}\in
D^m\:|\:\phi(\tup{x})<\infty\}$ and
any set $R\subseteq D^m$ can be represented
by $\phi_R:D^m\to\{0,\infty\}$ defined by $\phi_R(\tup{x})=0$ for $\tup{x}\in R$ and
$\phi_R(\tup{x})=\infty$ for $\tup{x}\not\in R$. Consequently, we shall use both definitions
interchangeably. 

}
\end{definition}

Let $\qq=\mathbb{Q}\cup\{\infty\}$ denote the set of rational numbers with (positive) infinity.

\begin{definition}\label{def:wrel}
An $m$-ary \emph{weighted relation} over $D$ is any mapping $\phi:D^m\to\qq$. We
write $ar(\phi)=m$ for the arity of $\phi$.
We denote by $\wrel_D$ the set of all weighted relations on $D$.
\end{definition}

For any $m$-ary weighted relation $\phi\in\wrel_D$, we denote by
$\feas(\phi)=\{\tup{x}\in D^m\:|\:\phi(\tup{x})<\infty\}\in\rel_D$ the
underlying $m$-ary \emph{feasibility relation}, and by
$\opt(\phi)=\{\tup{x}\in\feas(\phi)\:|\:\forall\tup{y}\in D^m:
\phi(\tup{x})\leq\phi(\tup{y})\}\in\rel_D$ the $m$-ary \emph{optimality relation},
which contains the tuples on which $\phi$ is minimised.
A weighted relation $\phi:D^m\to\qq$ is called \emph{finite-valued} if
$\feas(\phi)=D^m$.

\begin{definition}
Let $V=\{x_1,\ldots, x_n\}$ be a set of variables. A \emph{valued constraint} over $V$ is an expression
of the form $\phi(\tup{x})$ where $\phi\in \wrel_D$ and $\tup{x}\in V^{ar(\phi)}$. The number $m$ is called the \emph{arity} of the constraint,
the weighted relation $\phi$ is called the \emph{constraint weighted relation},
and the tuple $\tup{x}$ the \emph{scope} of the constraint.
\end{definition}

We call $D$ the \emph{domain}, the elements of $D$ \emph{labels}
and say that weighted relations take \emph{values}.

\begin{definition}
An instance of the \emph{valued constraint satisfaction problem} (VCSP) is specified
by a finite set $V=\{x_1,\ldots,x_n\}$ of variables, a finite set $D$ of labels,
and an \emph{objective function} $\inst$
expressed as follows:
\begin{equation}
\inst(x_1,\ldots, x_n)=\sum_{i=1}^q{\phi_i(\tup{x}_i)}\,,
\label{eq:sepfun}
\end{equation}
where each $\phi_i(\tup{x}_i)$, $1\le i\le q$, is a valued constraint over $V$.
Each constraint can appear multiple times in $\inst$.
The goal is to find an \emph{assignment} (or \emph{solution}) of labels to the
variables minimising $\inst$.
\end{definition}

A solution is called \emph{feasible} (or \emph{satisfying}) if it is of finite value. A VCSP instance
$I$ is called \emph{satisfiable} if there is a feasible solution to $I$.
CSPs are a special case of VCSPs with (unweighted) relations with the goal to
determine the existence of a feasible solution.
\begin{example} \label{ex:maxcut} 
In the \textsc{Min-UnCut} problem the goal is to find
a partition of the vertices of a given graph into two parts so that the number
of edges inside the two partitions is minimised.
For a graph $(V,E)$ with $V=\{x_1,\ldots,x_n\}$, this NP-hard problem can be expressed
as the VCSP instance $\inst(x_1,\ldots,x_n)=\sum_{(i,j)\in E} \phi_{\sf
xor}(x_i,x_j)$ over the Boolean domain $D=\{0,1\}$, where $\phi_{\sf
xor}:\{0,1\}^2\to\qq$ is defined by $\phi_{\sf
xor}(x,y)=1$ if $x=y$ and $\phi_{\sf xor}(x,y)=0$ if $x\neq y$.
\end{example}

\begin{definition}
Any set $\Delta\subseteq\rel_D$ is called a \emph{constraint language} over $D$.
Any set $\Gamma\subseteq\wrel_D$ is called a \emph{valued constraint
language} over $D$. We
denote by $\VCSP(\Gamma)$ the class of all VCSP instances in which the
constraint weighted relations are all contained in $\Gamma$. 
\end{definition}

For a constraint language $\Delta$, we denote by $\CSP(\Delta)$ the class
$\VCSP(\Delta)$ to emphasise the fact that there is no optimisation involved.

\begin{definition}
A valued constraint language $\Gamma$ is called \emph{tractable} if
$\VCSP(\Gamma')$ can be solved (to optimality) in
polynomial time for every finite subset $\Gamma'\subseteq\Gamma$, and $\Gamma$
is called \emph{intractable} if $\VCSP(\Gamma')$ is NP-hard for some finite
$\Gamma'\subseteq\Gamma$.
\end{definition}

Example~\ref{ex:maxcut} shows that the valued constraint language
$\{\phi_{\sf xor}\}$ is intractable. 

\subsection{Operations and Clones}

We recall some basic terminology from universal algebra.
Given an $m$-tuple $\tup{x}\in D^m$, we denote its $i$th entry by $\tup{x}[i]$ for $1\leq i\leq m$.
Any mapping $f:D^k\rightarrow D$ is called a $k$-ary \emph{operation}; $f$ is
called \emph{conservative} if $f(x_1,\ldots,x_k)\in\{x_1,\ldots,x_k\}$ and
\emph{idempotent} if $f(x,\ldots,x)=x$.
We will apply a $k$-ary operation $f$ to $k$ $m$-tuples
$\tup{x_1},\ldots,\tup{x_k}\in D^m$ coordinatewise, that is, \begin{equation}
f(\tup{x_1},\ldots,\tup{x_k})=(f(\tup{x_1}[1],\ldots,\tup{x_k}[1]),\ldots,f(\tup{x_1}[m],\ldots,\tup{x_k}[m]))\,.
\end{equation}

\begin{definition} \label{def:pol}
Let $\phi$ be an $m$-ary weighted relation on $D$. A $k$-ary operation $f$ on
$D$ is a \emph{polymorphism} of $\phi$ if,
for any $\tup{x_1},\ldots,\tup{x_k} \in D^m$ with
$\tup{x_i}\in\rdom(\phi)$ for all $1\leq i\leq k$,
we have that $f(\tup{x_1},\ldots,\tup{x_k})\in\rdom(\phi)$.

For any valued constraint language $\Gamma$ over a set $D$,
we denote by $\pol(\Gamma)$ the set of all operations on $D$ which are polymorphisms of all 
$\phi \in \Gamma$. We write $\pol(\phi)$ for $\pol(\{\phi\})$.
\end{definition}

A $k$-ary \emph{projection} is an operation of the form
$\proj^{(k)}_i(x_1,\ldots,x_k)=x_i$ for some $1\leq i\leq k$. Projections are
polymorphisms of all valued constraint languages. 

The \emph{composition} of a $k$-ary operation $f:D^k\rightarrow D$ with $k$
$\ell$-ary operations $g_i:D^\ell\rightarrow D$ for $1\leq i\leq k$ is the
$\ell$-ary function $f[g_1,\ldots,g_k]:D^\ell\to D$ defined by
\begin{equation}
f[g_1,\ldots,g_k](x_1,\ldots,x_\ell)=f(g_1(x_1,\ldots,x_\ell),\ldots,g_k(x_1,\ldots,x_\ell))\,.
\end{equation}

We denote by $\ops_D$ the set of all finitary operations on $D$ and by
$\ops_D^{(k)}$ the $k$-ary operations in $\ops_D$. 

A \emph{clone} of operations, $C\subseteq\ops_D$, is a set of operations on $D$ that contains
all projections and is closed under composition. 
It is easy to show that $\pol(\Gamma)$ is a clone for any valued constraint language $\Gamma$.

\begin{definition} \label{def:wop}
A $k$-ary \emph{fractional operation} $\omega$ 
is a probability distribution over $\ops_D^{(k)}$.
We define 
$\supp(\omega)=\{f\in \ops_D^{(k)} \mid \omega(f)>0\}$.
\end{definition}

\begin{definition} \label{def:wp} 
Let $\phi$ be an $m$-ary weighted relation 
on $D$ and let $\omega$ be a $k$-ary fractional operation 
on $D$.
We call $\omega$ a \emph{fractional polymorphism}
of $\phi$ if $\supp(\omega)\subseteq\pol(\phi)$ and for any
$\tup{x}_1,\ldots,\tup{x}_k \in D^m$ with
$\tup{x}_i\in\rdom(\phi)$ 
for all $1\leq i\leq k$, we have 
\begin{equation}
\E_{f\sim \omega}[\phi(f(\vec{x_1},\ldots,\vec{x_k}))]\ \le\
\avg\{\phi(\vec{x_1}),\ldots,\phi(\vec{x_k})\}\,.
\label{eq:wpol}
\end{equation}
We also say that $\phi$ is \emph{improved} by $\omega$.
\end{definition}

\begin{definition}
For any valued constraint language $\Gamma \subseteq \wrel_D$, we define $\fpol(\Gamma)$ to be the set of all
fractional operations
that are fractional polymorphisms of all weighted
relations $\phi \in \Gamma$. We write $\fpol(\phi)$ for $\fpol(\{\phi\})$.
\end{definition}

\begin{example}
A valued constraint language on domain $\{0,1\}$ is called \emph{submodular} if
it has the fractional polymorphism $\omega$ defined by
$\omega(\min)=\omega(\max)=\frac{1}{2}$, where $\min$ and $\max$ are the two
binary operations that return the smaller and larger of its two arguments
respectively with respect to the usual order $0<1$. 
\end{example}

\begin{definition}
Let $\Gamma$ be a valued constraint language on $D$. We define
\begin{equation}
\supp(\Gamma)\ =\ \bigcup_{\omega \in \fpol(\Gamma)} \supp(\omega)\,.
\end{equation}
\end{definition}

\begin{lemma}\label{lem:suppclone}
Let $\Gamma$ be a valued constraint language of finite size.
Then, $\supp(\Gamma)$ is a clone.
\end{lemma}

We note that Lemma~\ref{lem:suppclone} has also been
observed in~\cite{Ochremiak14:algebraic} and in~\cite{fz15:galois}.

\begin{proof}
Observe that $\supp(\Gamma)$ contains all projections as
$\tau_k\in\fpol(\Gamma)$ for every $k\geq 1$, where $\tau_k$ is the fractional
operation defined by $\tau_k(\proj^{(k)}_i)=\frac{1}{k}$ for every $1\leq i\leq
k$. Thus we only need to show that $\supp(\Gamma)$ is closed under composition.

Since $\omega\in\supp(\Gamma)$ there is $k$-ary $\omega \in \fpol(\Gamma)$ with
$\omega(f) > 0$. Moreover, since $g_1,\ldots,g_k\in\supp(\Gamma)$, for every $1\leq i\leq k$
there is $\ell$-ary $\mu_i\in\supp(\Gamma)$ with $\mu_i(g_i)>0$.
We define an $\ell$-ary fractional operation 
\begin{align}
\omega'(p)\ =\ \Pr_{\substack{t \sim \omega\\ h_i \sim \mu_i}} \left[ t[h_1, \ldots, h_k] = p \right]\,.
\end{align}

Since $\omega(f)>0$ and $\mu_i(g_i)>0$ for all $1\leq i\leq k$, we have
$\omega'(f[g_1,\ldots,g_k])>0$. 
A straightforward verification shows that $\omega'\in\fpol(\Gamma)$.
Consequently, $f[g_1,\ldots,g_k]\in\supp(\Gamma)$.
\end{proof}

The following lemma is a generalisation
of~\cite[Lemma~5]{tz12:arxiv-dichotomy-v2} from arity one to arbitrary arity and
from finite-valued to valued 
constraint languages, but the proof is analogous. A special case has also been observed, in
the context of Min-Sol problems~\cite{Uppman13:icalp}, by Hannes
Uppman.\footnote{Private communication.}

\begin{lemma}\label{lem:killing}
Let $\Gamma$ be a valued constraint language of finite size on a domain $D$ and let $f \in  \pol(\Gamma)$.
Then, $f \in \supp(\Gamma)$ if, and only if, $f \in \pol(\opt(I))$ for all instances $I$ of VCSP$(\Gamma)$.
\end{lemma}

\begin{proof}
The operation $f$ is in $\supp(\Gamma)$ if, and only if, there exists a fractional polymorphism
$\omega$ with $f \in \supp(\omega)$.
This is the case if, and only if, the following system of linear inequalities in 
the variables $\omega(g)$ for $g \in \pol(\Gamma)$ is satisfiable:
\begin{align}
  \sum_{g\in\pol(\Gamma)} \omega(g) \phi(f(\tup{x}_1,\dots,\tup{x}_k)) & \leq 
  \avg\{\phi(\tup{x}_1),\dots,\phi(\tup{x}_k)\}, \quad \forall \phi \in \Gamma, \tup{x}_i \in \feas(\phi), \notag \\
  \sum_{g\in\pol(\Gamma)} \omega(g) & = 1, \notag \\ 
  \omega(f) & > 0, \notag \\
  \omega(g) & \geq 0, \quad \forall g\in\pol(\Gamma).
  \label{eq:killing1}
\end{align}

By Farkas' lemma, the system (\ref{eq:killing1}) is unsatisfiable if, and only if, the following
system in variables $z(\phi,\tup{x}_1,\dots,\tup{x}_k)$, for $\phi \in \Gamma, \tup{x}_i \in \feas(\phi)$,
is satisfiable:
\begin{align}
  \sum_{\phi \in \Gamma, \tup{x}_i \in \feas(\phi)}
  z(\phi,\tup{x}_1, \dots, \tup{x}_k) \left( \avg\{\phi(\tup{x}_1),\dots,\phi(\tup{x}_k)\} 
  -\phi(g(\tup{x}_1,\dots,\tup{x}_k)) \right)  
  & \leq 0,
  \quad \forall g \in \pol(\Gamma), \notag \\
  \sum_{\phi \in \Gamma, \tup{x}_i \in \feas(\phi)}
  z(\phi,\tup{x}_1,\dots,\tup{x}_k) \left( \avg\{\phi(\tup{x}_1),\dots,\phi(\tup{x}_k)\} 
  - \phi(f(\tup{x}_1, \dots, \tup{x}_k)) \right)
  & < 0, \notag \\
  z(\phi, \tup{x}_1, \dots, \tup{x}_k) & \geq 0, \quad \forall \phi \in \Gamma, \tup{x}_i \in \feas(\phi).
  \label{eq:killing2}
\end{align}

First, assume that $f \not\in \supp(\Gamma)$ so that (\ref{eq:killing2}) has a feasible solution $z$.
Note that by scaling we may assume that $z$ is integral.
Then, $z$ can then be interpreted as an instance $I_f$ of VCSP$(\Gamma)$
in which we take as variables the $k$-tuples of $D$, $V = D^k$, and let
\[
I_f(\tup{x}) = \sum_{\phi \in \Gamma, \tup{x}_i \in \feas(\phi)}
z(\phi, \tup{x}_1, \dots, \tup{x}_k)
\phi((\tup{x}_1[1],\dots,\tup{x}_k[1]), \dots, (\tup{x}_1[ar(\phi)],\dots,\tup{x}_k[ar(\phi)])),
\]
where $\tup{x}$ is a list of the variables in $V$, and
the multiplication by $z$ is represented as taking the corresponding constraint with multiplicity $z$.
According to (\ref{eq:killing2}), 
any projection $\proj^{(k)}_i : D^k \to D$, $\proj^{(k)}_i(x_1,\dots,x_k) = x_i$ is an optimal assignment to $I_f$.
Interpreted as tuples, we therefore have $\proj^{(k)}_i \in \opt(I)$ for $1 \leq i \leq k$.
On the other hand, (\ref{eq:killing2}) states that $f$ is not an optimal assignment,
so $f(\proj^{(k)}_1, \dots, \proj^{(k)}_k) \not\in \opt(I_f)$.
In other words, $f \not\in \pol(\opt(I_f))$.

For the opposite direction, assume that $f \in \supp(\Gamma)$,
so that (\ref{eq:killing2}) is unsatisfiable.
Let $I$ be an arbitrary instance of VCSP$(\Gamma)$, and let
$\sigma_1, \dots, \sigma_k \in \opt(I)$ be $k$ optimal solutions to $I$.
Construct an instance $Z$ of VCSP$(\Gamma)$ with variables $D^k$ by replacing each
valued constraint $\phi_i(\tup{x}_i)$ in $I$ by $\phi_i(\sigma_1(\tup{x}_i),\dots,\sigma_k(\tup{x}_i))$,
in $Z$,
where $(\sigma_1(\tup{x}_i),\dots,\sigma_k(\tup{x}_i))$ is a tuple of variables in $(D^k)^{ar(\phi_i)}$.
Now, if $f$ were not an optimal solution to $Z$, then $Z$ would be a solution to (\ref{eq:killing2}),
a contradiction.
Hence $f \in \pol(\opt(I))$.
Since $I$ and $\sigma_i$ were chosen arbitrarily, this establishes the lemma.
\end{proof}

\subsection{Cores and Constants}

\begin{definition}
Let $\Gamma$ be a valued constraint language with domain $D$ and let $S \subseteq D$.
The \emph{sub-language $\Gamma[S]$ of $\Gamma$ induced by $S$} is the valued constraint
language defined on domain $S$ and containing the restriction of every weighted relation
$\phi\in\Gamma$ onto $S$.
\end{definition}

\begin{definition}\label{def:core}
A valued constraint language $\Gamma$ is \emph{a core} if all unary operations in
$\supp(\Gamma)$ are bijections.
A valued constraint language $\Gamma'$ is a \emph{core of $\Gamma$} if $\Gamma'$ is a
core and $\Gamma' = \Gamma[f(D)]$ for some $f \in \supp(\omega)$ with $\omega$
a unary fractional polymorphism of $\Gamma$.
\end{definition}

The following lemma implies that when studying the computational complexity of a
valued constraint language $\Gamma$‚ way may assume that $\Gamma$ is a core.

\begin{lemma}\label{lem:core}
  Let $\Gamma$ be a valued constraint language and $\Gamma'$ a core of $\Gamma$.
  Then, for all instances $I$ of VCSP$(\Gamma)$ and $I'$ of VCSP$(\Gamma')$,
  where $I'$ is obtained from $I$ by
  substituting each function in $\Gamma$ for its restriction in $\Gamma'$,
  the optimum of $I$ and $I'$ coincide.
\end{lemma}

A special case of Lemma~\ref{lem:core} for finite-valued constraint languages
was proved by the authors in~\cite{tz13:stoc}. Lemma~\ref{lem:core} has also been
observed in~\cite{Ochremiak14:algebraic} and in another recent paper of the
authors~\cite{tz15:necessary}.

\begin{proof}
By definition, $\Gamma' = \Gamma[f(D)]$,
where $D$ is the domain of $\Gamma$ and $f \in \supp(\omega)$
for some unary fractional polymorphism $\omega$.
Assume that $I$ is satisfiable, and let $\sigma$ be an optimal assignment to $I$.
Now $f \circ \sigma$ is a satisfying assignment to $I'$, and by Lemma~\ref{lem:killing}, 
$f \circ \sigma$ is also an optimal assignment to $I$.
Conversely, any satisfying assignment to $I'$ is a satisfying assignment to $I$
of the same value.
\end{proof}

Let $\mathcal{C}_D = \{ \{(d)\} \mid d \in D \}$ be the set of constant unary relations on the set $D$.

\begin{lemma}[\hspace*{-0.4em}\cite{Ochremiak14:algebraic}]\label{lem:constants}
  Let $\Gamma$ be a core valued constraint language.
  The problems VCSP$(\Gamma)$ and VCSP$({\Gamma \cup \mathcal{C}_D})$
  are polynomial-time equivalent.
\end{lemma}

A special case of Lemma~\ref{lem:constants} for finite-valued constraint
languages was proved by the authors in~\cite{tz13:stoc}, building
on~\cite{hkp14:sicomp}, and Lemma~\ref{lem:constants} can be proved similarly;
we refer the reader to~\cite{Ochremiak14:algebraic}.

\section{Sherali-Adams Relaxations and Valued Relational Width}

In this section, we state and prove our main result on the applicability of 
Sherali-Adams relaxations to VCSPs.
First, we define some notions concerning \emph{bounded relational width}
which is the basis for our proof.

We write $(S,C)$ for (valued) constraints that involve (unweighted)
relations, where $S$ is the scope and $C$ is the constraint relation. 
For a tuple $\tup{x} \in D^S$, we denote by $\pi_{S'}(\tup{x})$ its projection onto $S' \subseteq S$.
For a constraint $(S,C)$, we define $\pi_{S'}(C) = \{ \pi_{S'}(\tup{x}) \mid \tup{x} \in C \}$.

Let $1 \leq k \leq \ell$ be integers.
The following definition is equivalent\footnote{The two requirements
in~\cite{Barto14:jloc} are: for every $S\subseteq V$ with $|S|\leq\ell$ we have
$S\subseteq S_i$ for some $1\leq i\leq q$; and for every set $W\subseteq V$ with
$|W|\leq k$ and every $1\leq i,j\leq q$ with $W\subseteq S_i$ and $W\subseteq
S_j$ we have $\pi_{W}(C_i)=\pi_{W}(C_j)$.} to the definition of $(k,\ell)$-minimality for CSP instances
given in~\cite{Barto14:jloc}.

\begin{definition}
A CSP-instance $J = (V, D, \{(S_i, C_i)\}_{i=1}^q)$ is said to be \emph{$(k,\ell)$-minimal} if:
\begin{itemize}
\item
For every $S \subseteq V$, $\left|S\right| \leq \ell$, there exists $1 \leq i \leq q$ such that $S = S_i$.
\item
For every $i, j \in \left[q\right]$ such that $\left| S_j \right| \leq k$ and $S_j \subseteq S_i$,
$C_j = \pi_{S_j}(C_i)$.
\end{itemize}
\end{definition}

There is a straightforward polynomial-time algorithm for finding an equivalent $(k,\ell)$-minimal instance~\cite{Barto14:jloc}.
This leads to notion of \emph{relational width}:

\begin{definition}
A constraint language $\Delta$ has relational width $(k,\ell)$ if, for every instance
$J\in\CSP(\Delta)$, an equivalent $(k,\ell)$-minimal instance is non-empty if,
and only if,
$J$ has a solution.
\end{definition}

A $k$-ary idempotent operation $f:D^k\to D$ is called a 
\emph{weak near-unanimity} (WNU) operation if, for all $x,y\in D$,
\[
f(y,x,x,\ldots,x)=f(x,y,x,x,\ldots,x)=f(x,x,\ldots,x,y)\,.
\]

\begin{definition}[BWC]
We say that a clone of operations satisfies the \emph{bounded width condition (BWC)}
if it contains WNU operations of all but finitely many arities.
\end{definition}

\begin{theorem}[\hspace*{-0.3em}\cite{Barto14:jacm,LaroseZadori07:au}]\label{thm:boundedrelwidth}
Let $\Delta$ be a constraint language of finite size containing all constant unary relations.
Then, $\Delta$ has bounded relational width if, and only if, $\pol(\Delta)$ satisfies the BWC.
\end{theorem}

\begin{theorem}[\hspace*{-0.3em}\cite{Barto14:jloc}]\label{thm:libor23}
Let $\Delta$ be a constraint language.
If $\Delta$ has bounded relational width, then it has relational width $(2,3)$.
\end{theorem}

Let $I(x_1,\dots,x_n) = \sum_{i=1}^q \phi_i(S_i)$ be an instance of the VCSP,
where $S_i \subseteq V = \{x_1, \dots, x_n\}$ and $\phi_i \colon
D^{\left|S_i\right|} \to \qq$.
First, we make sure that every non-empty $S \subseteq V$ with $|S| \leq \ell$
appears in some term $\phi_i(S)$,
possibly by adding constant-0 weighted relations.
The Sherali-Adams~\cite{Sherali1990} linear programming relaxation with
parameters $(k,\ell)$ is defined as follows.
The variables are $\lambda_{i}(\tup{s})$ for every $i \in \left[q\right]$ and tuple $\tup{s} \in D^{S_i}$.

\begin{align}
\min \sum_{i = 1}^q \sum_{\tup{s} \in \feas(\phi_i)} \lambda_{i}(\tup{s}) \phi_i(\tup{s}) \nonumber\\
\lambda_{j}(\tup{t}) &= \sum_{\tup{s} \in D^{S_i}, \pi_{S_j}(\tup{s}) = \tup{t}} \lambda_i(\tup{s}) &
   \forall i, j \in \left[q\right] \text{ s.t. } S_j \subseteq S_i, \left|S_j\right| \leq k, \tup{t} \in D^{S_j} \label{sa:marginal} \\
\sum_{\tup{s} \in D^{S_i}} \lambda_{i}(\tup{s}) &= 1 &
   \forall i\in\left[q\right] \label{sa:sum1} \\
\lambda_{i}(\tup{s}) &= 0 &
   \forall i \in \left[q\right], \tup{s} \not\in \feas(\phi_i) \label{sa:infeas} \\
\lambda_{i}(\tup{s}) &\geq 0 &
   \forall i \in \left[q\right], \tup{s} \in D^{S_i} \label{sa:nonnegative}
\end{align}

The SA$(k,\ell)$ optimum is always less than or equal to the VCSP optimum,
hence the program is a relaxation.
In anticipation of our main theorem, we make the following definition.

\begin{definition}
A valued constraint language $\Gamma$ has \emph{valued relational width
$(k,\ell)$} if,
for every instance $I$ of $\VCSP(\Gamma)$,
if the SA$(k,\ell)$-relaxation of $I$ has a feasible solution,
then its optimum coincides with the optimum of $I$.
\end{definition}

For a feasible solution $\lambda$ of SA$(k,\ell)$, let
$\supp(\lambda_i) = \{ \tup{s} \in D^{S_i} \mid \lambda_i(\tup{s}) > 0 \}$.

\begin{lemma}\label{lem:fullsupport}
Let $I$ be an instance of VCSP$(\Gamma)$.
Assume that SA$(k,\ell)$ for $I$ is feasible.
Then, there exists an optimal solution $\lambda^*$ to SA$(k,\ell)$ such that,
for every $i$,
$\supp(\lambda^*_i)$ is closed under 
every operation in $\supp(\Gamma)$.
\end{lemma}
\begin{proof}
Let $\omega$ be an arbitrary $m$-ary fractional polymorphism of $\Gamma$,
and let $\lambda$ be any feasible solution $\lambda$ to SA$(k,\ell)$. Define
$\lambda^\omega$ by
\[
\lambda^\omega_i(\tup{s}) = \Pr_{\substack{f \sim \omega \\ \tup{s}_1,\dots,\tup{s}_m \sim \lambda_i}}
[ f(\tup{s}_1,\dots,\tup{s}_m) = \tup{s}].
\]
We show that $\lambda^\omega$ is a feasible solution to SA$(k,\ell)$, and that if $\lambda$ is optimal,
then so is $\lambda^\omega$.

Clearly $\lambda^\omega_i$ is a probability distribution for each $i \in \left[q\right]$,
so (\ref{sa:sum1}) and (\ref{sa:nonnegative}) hold.
Since $\omega$ is a fractional polymorphism of $\Gamma$,
we have $\tup{s} \in \feas(\phi_i)$ 
for any choice of $f \in \supp(\omega)$ and 
$\tup{s}_1, \dots, \tup{s}_m \in \supp(\lambda_i)$.
Hence, $\lambda^\omega_i(\tup{s}) = 0$ for $\tup{s} \not\in \feas(\phi_i)$,
so (\ref{sa:infeas}) holds.

Finally, let $j \in \left[q\right]$ be such that $S_j \subseteq S_i$, $\left|S_j\right| \leq k$,
and let $\tup{t} \in D^{S_j}$. Then,
\begin{align*}
\sum_{\tup{s} \in D^{S_i}, \pi_{S_j}(\tup{s}) = \tup{t}} \lambda^\omega_i(\tup{s}) 
&=
\sum_{\tup{s} \in D^{S_i}, \pi_{S_j}(\tup{s}) = \tup{t}} \Pr_{\substack{f \sim \omega \\ \tup{s}_1,\dots,\tup{s}_m \sim \lambda_i}}
[ f(\tup{s}_1,\dots,\tup{s}_m) = \tup{s}]\\
&=
\Pr_{\substack{f \sim \omega \\ \tup{s}_1,\dots,\tup{s}_m \sim \lambda_i}}
[ \pi_{S_j}(f(\tup{s}_1,\dots,\tup{s}_m)) = \tup{t}]\\
&=
\Pr_{\substack{f \sim \omega \\ \tup{s}_1,\dots,\tup{s}_m \sim \lambda_i}}
[ \tup{t}_1 = \pi_{S_j}(\tup{s}_1) \wedge \dots \wedge \tup{t}_m = \pi_{S_j}(\tup{s}_1) \wedge 
f(\tup{t}_1,\dots,\tup{t}_m) = \tup{t}]\\
&=
\Pr_{\substack{f \sim \omega \\ \tup{t}_1,\dots,\tup{t}_m \sim \lambda_j}}
[ f(\tup{t}_1,\dots,\tup{t}_m) = \tup{t}]\\
&=
\lambda^\omega_j(\tup{t}),
\end{align*}
where, in the penultimate equality, we have used the fact that
(\ref{sa:marginal}) can be read as 
$\lambda_j(\tup{t}) = \Pr_{\tup{s} \sim \lambda_i} \left[\pi_{S_j}(\tup{s}) = \tup{t}\right]$.
It follows that (\ref{sa:marginal}) also holds for $\lambda^\omega$, so $\lambda^\omega$ is feasible.

For each $i \in \left[q\right]$, we have
\begin{align*}
\sum_{\tup{s}\in \feas(\phi_i)} \lambda_i(\tup{s})\phi_i(\tup{s}) &=
\E_{\tup{s} \sim \lambda_i} \phi_i(\tup{s}) = \E_{\tup{s}_1,\dots,\tup{s}_m \sim \lambda_i} \sum_{j=1}^m \phi_i(\tup{s}_j)\\
&\geq \E_{\substack{f \sim \omega\\ \tup{s}_1,\dots,\tup{s}_m \sim \lambda_i}} \phi_i(f(\tup{s}_1,\dots,\tup{s}_m))\\
&= \sum_{s \in \feas(\phi_i)} \Big( \Pr_{\substack{f \sim \omega\\ \tup{s}_1,\dots,\tup{s}_m \sim \lambda_i}} \left[f(\tup{s}_1,\dots,\tup{s}_m) = \tup{s}\right] \Big) \phi_i(\tup{s})\\
&= \sum_{s \in \feas(\phi_i)} \lambda^\omega_i(\tup{s}) \phi_i(\tup{s}).
\end{align*}
Therefore, if $\lambda$ is optimal, then $\lambda^\omega$ must also be optimal.

Now assume that $\lambda$ is an optimal solution and that $\supp(\lambda)$ is not
closed under some operation $f \in \supp(\omega)$ for $\omega \in \fpol(\Gamma)$,
i.e., for some $\tup{s_1}, \dots, \tup{s_m} \in \supp(\lambda)$, we have
$f(\tup{s}_1,\dots,\tup{s}_m) \not\in \supp(\lambda)$.
But note that  $f(\tup{s}_1,\dots,\tup{s}_m) \in \supp(\lambda^\omega_i)$.
Therefore, $\lambda' = \frac{1}{2}(\lambda+\lambda^\omega)$ is an optimal solution
such that $\supp(\lambda_i) \subsetneq \supp(\lambda'_i) \subseteq D^{S_i}$.
For each $i \in \left[q\right]$, $D^{S_i}$ is finite.
Hence, by repeating this procedure, we obtain a sequence of optimal solutions with
strictly increasing support until, after a finite number of steps, we obtain a $\lambda^*$
that is closed under
every operation in $\supp(\Gamma)$.
\end{proof}

\begin{theorem}\label{thm:maincons}
Let $\Gamma$ be a valued constraint language of finite size containing all constant unary relations.
If $\supp(\Gamma)$ satisfies the BWC, then $\Gamma$ has valued relational width $(2,3)$.
\end{theorem}

\begin{proof}
Let $I$ be an instance of VCSP$(\Gamma)$.
The dual of the SA$(k,\ell)$ relaxation can be written in the following form.
The variables are $z_i$ for $i \in \left[q\right]$ and $y_{j, \tup{t}, i}$ for $i, j \in \left[q\right]$ such that $S_j \subseteq S_i$, $\left|S_j\right| \leq k$, and $\tup{t} \in D^{S_j}$.
\begin{align}
\max \sum_{i=1}^q z_i \nonumber\\
z_i &\leq \phi_i(\tup{s}) + \sum_{j \in \left[q\right], S_j \subseteq S_i} y_{j, \pi_{S_j}(\tup{s}), i} - \sum_{j \in \left[q\right], S_i \subseteq S_{j}} y_{i, \tup{s}, j}  & \forall i \in \left[q\right], \left|S_i\right| \leq k, \tup{s} \in \feas(\phi_i) \label{dual:smallk} \\
z_i &\leq \phi_i(\tup{s}) + \sum_{\substack{j \in \left[q\right], S_j \subseteq S_i \\ \left|S_j\right| \leq k}} y_{j,\pi_{S_j}(\tup{s}),i} & \forall i \in \left[q\right], |S_i| > k, \tup{s} \in \feas(\phi_i) \label{dual:largek}
\end{align}

It is clear that if $I$ has a feasible solution, then so does the SA$(k,\ell)$ primal.
Assume that the SA$(2,3)$-relaxation has a feasible solution.
By Lemma~\ref{lem:fullsupport}, there exists an optimal primal solution $\lambda^*$ such that,
for every $i \in \left[q\right]$,
$\supp(\lambda^*_i)$ is closed under $\supp(\Gamma)$.
Let $y^*$, $z^*$ be an optimal dual solution.

Let $\Delta = \{C_i\}_{i=1}^q \cup \{ \mathcal{C}_D \}$, where $C_i = \supp(\lambda^*_i)$, and
consider the instance $J = (V,D,\{(S_i,C_i)\}_{i=1}^q)$ of CSP$(\Delta)$.
We make the following observations:
\begin{enumerate}
\item
By construction of $\lambda^*$, $\supp(\Gamma) \subseteq \pol(\Delta)$, so $\Delta$
contains all constant unary relations and
satisfies the BWC.
By Theorems~\ref{thm:boundedrelwidth} and~\ref{thm:libor23}, the language $\Delta$ has relational width $(2,3)$. 
\item
The constraints (\ref{sa:marginal}) say that if $i, j \in \left[q\right]$,
$\left|S_j\right| \leq 2$ and $S_j \subseteq S_i$, then $\lambda^*_j(\tup{t}) >
0$ (i.e., $\tup{t} \in C_j$) if, and only if, $\sum_{\tup{s} \in D^{S_i}, \pi_{S_j}(\tup{s}) = \tup{t}} \lambda^*_i(\tup{s}) > 0$ (i.e., $\tup{t} \in \pi_{S_j}(C_i)$). In other words, $J$ is $(2,3)$-minimal.
\end{enumerate}

These two observations imply that $J$ has a satisfying assignment $\sigma \colon V \to D$.

By complementary slackness, since $\lambda^*_i(\sigma(S_i)) > 0$ for every $i \in \left[q\right]$, 
we must have equality in the corresponding rows in the dual indexed by $i$ and $\sigma(S_i)$.
We sum these rows over $i$:
\begin{equation}\label{eq:dualsum}
\sum_{i=1}^q z^*_i\ =\ \sum_{i=1}^q \phi_i(\sigma(S_i)) +
\Big( \sum_{i=1}^q \sum_{\substack{j \in \left[q\right], S_j \subseteq S_i \\ \left|S_j\right| \leq 2}} y^*_{j,\pi_{S_j}(\sigma(S_i)),i} -
 \sum_{\substack{i \in \left[q\right]\\ \left|S_i\right| \leq 2}} \sum_{j\in\left[q\right], S_i \subseteq S_j} y^*_{i,\sigma(S_i),j}
 \Big).
\end{equation}

By noting that $\pi_{S_j}(\sigma(S_i)) = \sigma(S_j)$, we can rewrite the expression in parenthesis on
the right-hand side of (\ref{eq:dualsum}) as:
\begin{equation}\label{eq:is0}
 \sum_{\substack{i, j \in \left[q\right], S_j \subseteq S_i\\ \left|S_j\right| \leq 2}} y^*_{j, \sigma(S_j), i} -
 \sum_{\substack{i, j \in \left[q\right], S_j \subseteq S_i\\ \left|S_i\right| \leq 2}} y^*_{j, \sigma(S_j), i}\ =\ 0.
\end{equation}

Therefore,
\[
\sum_{i=1}^q \sum_{\tup{s} \in \feas(\phi_i)} \lambda^*_i(\tup{s}) \phi_i(\tup{s})\ =\ 
\sum_{i=1}^q z^*_i\ =\ \sum_{i=1}^q \phi_i(\sigma(S_i)),
\]
where the first equality follows by strong LP-duality, and the second by (\ref{eq:dualsum}) and
(\ref{eq:is0}).

Since $I$ was an arbitrary instance of VCSP$(\Gamma)$,
we conclude that $\Gamma$ has valued relational width $(2,3)$.
\end{proof}

\section{Generalisations of Known Tractable Languages}

In this section, we give some applications of Theorem~\ref{thm:maincons}.
Firstly, we show that the BWC is preserved by going to a core and the
addition of constant unary relations.

\begin{lemma}\label{lem:coreBWC}
Let $\Gamma$ be a valued constraint language of finite size on domain $D$ 
and $\Gamma'$ a core of $\Gamma$ on domain $D' \subseteq D$.
Then,
$\supp(\Gamma)$ satisfies the BWC if, and only if, $\supp(\Gamma' \cup \mathcal{C}_{D'})$
satisfies the BWC.
\end{lemma}

\begin{proof}
Let $\mu$ be a unary fractional polymorphism of $\Gamma$ with an operation $g$
in its support such that $g(D) = D'$.
We begin by constructing a unary fractional polymorphism $\mu'$ of $\Gamma$ such that
\emph{every} operation in $\supp(\mu')$ has an image in $D'$.

We will use a technique for generating fractional polymorphisms described
in~\cite[Lemma 10]{ktz15:sicomp}.
It takes a fractional polymorphism, such as $\mu$, a set of
\emph{collections} $\mathbb{G}$,
which in our case will be the set of operations in the clone of $\supp(\mu)$, a set of \emph{good}
collections $\mathbb{G^*}$, which will be operations from $\mathbb{G}$ with an image in $D'$,
and an \emph{expansion operator} {\sf Exp} which assigns to every collection 
a probability distribution on $\mathbb{G}$.

The procedure starts by generating each collection $f \in \supp(\mu)$ with probability $\mu(f)$,
and subsequently the expansion operation {\sf Exp} maps 
$f \in \mathbb{G}$ to the probability distribution 
that assigns probability $\Pr_{h \sim \mu} [h \circ f = f']$ to each operation $f' \in \mathbb{G}$.
The expansion operator is required to be \emph{non-vanishing}, which means that
starting from any collection $f \in \mathbb{G}$, repeated expansion must assign
non-zero probability to a good collection in $\mathbb{G}^*$.
In our case, this is immediate, since starting from a collection $f$, the good collection
$g \circ f$ gets probability at least $\mu(g)$ which is non-zero by assumption.
By \cite[Lemma~10]{ktz15:sicomp}, it now follows that $\Gamma$ has a fractional
polymorphism $\mu'$ with $\supp(\mu') \subseteq \mathbb{G}^*$.
So every operation in $\supp(\mu')$ has an image in $D'$.

Now, we show that if $\supp(\Gamma)$ contains an $m$-ary WNU $t$, then
$\supp(\Gamma' \cup \mathcal{C}_{D'})$ also contains an $m$-ary WNU.
Let $\omega$ be a fractional polymorphism of $\Gamma$ with $t$ in its support.
Define $\omega'$ by $\omega'(f') = \Pr_{h \sim \mu', f \sim \omega} [h \circ f = f']$.
Then, $\omega'$ is a
fractional polymorphism of $\Gamma$ in which every operation has an image in $D'$,
so $\omega'$ is a fractional polymorphism of $\Gamma'$.
Furthermore, for any unary operation $h \in \supp(\mu')$, $h \circ t$ is again a WNU, so
$\supp(\Gamma')$ contains an $m$-ary WNU $t'$.
Next, let $h(x) = t'(x, \dots, x)$.
Since $\Gamma'$ is a core, the set of unary operations in $\supp(\Gamma')$ contains only 
bijections and is closed under composition (Lemma~\ref{lem:suppclone}).
It follows that $h$ has an inverse $h^{-1} \in \supp(\Gamma')$, and since $\supp(\Gamma')$
is a clone, $h^{-1} \circ t'$ is an idempotent WNU
in $\supp(\Gamma')$.
We conclude that $h^{-1} \circ t' \in \supp(\Gamma' \cup \{ \mathcal{C}_{D'} \})$.

For the opposite direction, let $t'$ be an $m$-ary WNU in $\supp(\Gamma' \cup \{ \mathcal{C}_{D'} \})$,
and let $\omega'$ be a fractional polymorphism of $\Gamma' \cup \{ \mathcal{C}_{D'} \}$
with $t'$ in its support.
Then, $\omega'$ is also a fractional polymorphism of $\Gamma'$.
Define $\omega$ by $\omega(f) = \Pr_{h \sim \mu', f' \sim \omega'} [f'[h, \dots, h] = f]$.
Then, $\omega$ is a fractional polymorphism of $\Gamma$, and, for every $h \in \supp(\mu')$,
the operation $t[h, \dots, h]$ is an $m$-ary WNU in $\supp(\omega)$.
We conclude that $t \in \supp(\Gamma)$,
which finishes the proof. 
\end{proof}

Hence the BWC guarantees valued relational width $(2,3)$ also for languages not
necessarily containing constant unary relations, as required by Theorem~\ref{thm:maincons}.

\begin{theorem}\label{thm:main}
Let $\Gamma$ be a valued constraint language of finite size.
If $\supp(\Gamma)$ satisfies the BWC, then $\Gamma$ has valued relational width
$(2,3)$.
\end{theorem}

\begin{proof}
Let $D$ be the domain of $\Gamma$, and $D' \subseteq D$ the domain of a core
$\Gamma'$ of $\Gamma$.
By Lemma~\ref{lem:coreBWC} and Theorem~\ref{thm:maincons}, the language 
$\Gamma' \cup \{ \mathcal{C}_{D'} \}$ has valued relational width $(2,3)$,
so clearly $\Gamma'$ has valued relational width $(2,3)$ as well.
Every feasible solution to the SA$(2,3)$-relaxation of an instance $I'$ of VCSP$(\Gamma')$
is also a feasible solution to the SA$(2,3)$-relaxation of the corresponding instance $I$ of
VCSP$(\Gamma)$.
The result now follows from Lemma~\ref{lem:core} as the optimum of $I'$ and $I$ coincide.
\end{proof}

Secondly, we show that for any VCSP instance over a language of valued
relational width $(2,3)$ we can not only compute the value of an optimal solution but
we can also find an optimal assignment in polynomial time.

\begin{proposition}\label{prop:self-reduce}
Let $\Gamma$ be a valued constraint language of finite size and $I$ an instance of VCSP$(\Gamma)$.
If $\supp(\Gamma)$ satisfies the BWC, then an optimal assignment to $I$ can be found in polynomial time.
\end{proposition}

\begin{proof}
Let $\Gamma'$ be a core of $\Gamma$ on domain $D'$, 
and let $\Gamma_c =  \Gamma' \cup \{ \mathcal{C}_{D'} \}$.
By Lemma~\ref{lem:coreBWC}, $\supp(\Gamma_c)$ satisfies the BWC,
so by Theorem~\ref{thm:maincons} we can obtain the optimum of $I$
by solving a linear programming relaxation.
Now, we can use self-reduction to obtain an optimal assignment.
It suffices to modify the instance $I$ to successively force each variable
to take on each value of $D'$. Whenever the optimum of the modified instance
matches that of the original instance, we can move on to assign the next variable.
This means that we need to solve at most $1+\left|V\right| \left|D'\right|$ linear
programming relaxations before finding an optimal assignment,
where $V$ is the set of variables of $I$.
\end{proof}

Finally, we show that testing for the BWC is a decidable problem.
We rely on the following result that was proved in~\cite{Kozik14:au}, and also follows
from results in~\cite{Barto14:jloc}.

\begin{theorem}[\hspace*{-0.3em}\cite{Kozik14:au}]\label{thm:34}
An idempotent clone $C$ of operations satisfies the BWC if, and only if, $C$
contains a ternary WNU $f$ and a $4$-ary WNU $g$ with $f(y,x,x)=g(y,x,x,x)$ for
all $x$ and $y$.
\end{theorem}

\begin{proposition}
Testing whether a valued constraint language of finite size
satisfies the BWC is decidable.
\end{proposition}

\begin{proof}
Let $\Gamma$ be a valued constraint language of finite size on domain $D$. Let
$\Gamma'$ be a core of $\Gamma$ defined on domain $D'\subseteq D$. Finding $D'$
and $\Gamma'$ can be done via linear
programming~\cite[Section~4]{tz12:arxiv-dichotomy-v2}. 
By Lemma~\ref{lem:coreBWC}, $\supp(\Gamma)$
satisfies the BWC if, and only if, $\supp(\Gamma'\cup \mathcal{C}_{D'})$ satisfies the
BWC. As constant unary relations enforce idempotency,
by Theorem~\ref{thm:34},
$\supp(\Gamma'\cup\mathcal{C}_{D'})$ satisfies the BWC if, and only if,
$\supp(\Gamma'\cup\mathcal{C}_{D'})$ contains a ternary WNU $f$ and a 4-ary 
WNU $g$ with $f(y,x,x)=g(y,x,x,x)$ for all $x$ and $y$. It is easy to write a
linear program that checks for this condition, as it has been
done in the context of finite-valued constraint languages~\cite[Section~4]{tz12:arxiv-dichotomy-v2}.
\end{proof}

\subsection{Tractable Languages}

Here we give some examples of previously studied valued constraint languages
and show that, as a corollary of Theorem~\ref{thm:main},
they all have valued relational width $(2,3)$.

\begin{example}\label{ex:maj}
Let $\omega$ be a ternary fractional operation defined by
$\omega(f)=\omega(g)=\omega(h)=\frac{1}{3}$ for some 
(not necessarily distinct) majority operations $f$,
$g$, and $h$. Cohen et al. proved the tractability of any language improved by
$\omega$ by a reduction to CSPs with a majority
polymorphism~\cite{Cohen06:complexitysoft}.
\end{example}

\begin{example}\label{ex:cohen}
Let $\omega$ be a ternary fractional operation defined by $\omega(f)=\frac{2}{3}$ and $\omega(g)=\frac{1}{3}$, 
where $f:\{0,1\}^3\to\{0,1\}$ is the Boolean majority operation and
$g:\{0,1\}^3\to\{0,1\}$ is the Boolean minority operation.
Cohen et al. proved the tractability of any language improved by $\omega$ by a
simple propagation algorithm~\cite{Cohen06:complexitysoft}.
\end{example}

\begin{example}\label{ex:mjn}
Generalising Example~\ref{ex:cohen} from
Boolean to arbitrary domains, let $\omega$ be a ternary fractional operation
such that $\omega(f)=\frac{1}{3}$,
$\omega(g)=\frac{1}{3}$, and $\omega(h)=\frac{1}{3}$ for some (not necessarily
distinct) majority operations $f$ and $g$, and a minority operation $h$; such an
$\omega$ is called an MJN.
Kolmogorov and \v{Z}ivn\'y proved the tractability of any
language improved by $\omega$ by a 3-consistency algorithm and a reduction, via
Example~\ref{ex:stp}, to submodular function minimisation~\cite{kz13:jacm}.
\end{example}

The following corollary of Theorem~\ref{thm:main} generalises Examples~\ref{ex:maj}-\ref{ex:mjn}.

\begin{corollary}\label{cor:maj}
Let $\Gamma$ be a valued constraint language of finite size such that $\supp(\Gamma)$
contains a majority operation.
Then, $\Gamma$ has valued relational width $(2,3)$.
\end{corollary}
\begin{proof}
Let $f$ be a majority operation in $\supp(\Gamma)$.
Then, for every $k \geq 3$, $f$ generates a WNU $g_k$ of arity $k$:
$g_k(x_1,\dots,x_k) = f(x_1,x_2,x_3)$.
By Lemma~\ref{lem:suppclone}, $\supp(\Gamma)$ is a clone,
so $g_k \in \supp(\Gamma)$ for all $k \geq 3$.
Therefore, $\supp(\Gamma)$ satisfies the BWC and
the result follows from Theorem~\ref{thm:main}.
\end{proof}

\begin{example}\label{ex:stp}
Let $\omega$ be a binary fractional operation defined by
$\omega(f)=\omega(g)=\frac{1}{2}$, where $f$ and $g$ are conservative and
commutative operations and $f(x,y)\neq g(x,y)$ for every $x$ and $y$; such an
$\omega$ is called a \emph{symmetric tournament pair} (STP). Cohen et al. proved
the tractability of any language improved by $\omega$ by a 3-consistency
algorithm and an ingenious reduction to submodular function
minimisation~\cite{Cohen08:Generalising}. Such languages were shown to be the
only tractable languages among conservative finite-valued constraint
languages~\cite{kz13:jacm}. 
\end{example}

The following corollary of Theorem~\ref{thm:main} generalises Example~\ref{ex:stp}.

\begin{corollary}\label{cor:stp}
Let $\Gamma$ be a valued constraint language of finite size such that $\supp(\Gamma)$
contains two symmetric tournament operations (that is, binary operations
$f$ and $g$ that are both conservative and commutative and $f(x,y)\neq g(x,y)$
for every $x$ and $y$).
Then, $\Gamma$ has valued relational width $(2,3)$.
\end{corollary}

\begin{proof}
It is straightforward to verify that $h(x,y,z)=f(f(g(x,y),g(x,z)),g(y,z))$ is a
majority operation, as observed in~\cite[Corollary~5.8]{Cohen08:Generalising}.
The claim then follows from Corollary~\ref{cor:maj}.
\end{proof}

\begin{example}\label{ex:tp}
Generalising Example~\ref{ex:stp},
let $\omega$ be a binary fractional operation defined by
$\omega(f)=\omega(g)=\frac{1}{2}$, where $f$ and $g$ are conservative and
commutative operations; such
an $\omega$ is called a \emph{tournament pair}. Cohen et al. proved the
tractability of any language improved by $\omega$ by a consistency-reduction
relying on Bulatov's result~\cite{Bulatov06:ja}, which in turn relies on
3-consistency, to the STP case from Example~\ref{ex:stp}~\cite{Cohen08:Generalising}.
\end{example}

The following corollary of Theorem~\ref{thm:main} generalises Example~\ref{ex:tp}.

\begin{corollary}\label{cor:tp}
Let $\Gamma$ be a valued constraint language of finite size such that $\supp(\Gamma)$
contains a tournament operation (that is, a binary conservative and commutative
operation).
Then, $\Gamma$ has valued relational width $(2,3)$.
\end{corollary}

\begin{proof} 
Let $f$ be a tournament operation from $\supp(\Gamma)$. We claim that $f$ is a
2-semilattice; that is, $f$ is idempotent, commutative, and satisfies the
restricted associativity law $f(x,f(x,y)) = f(f(x,x),y)$.  To see that, notice
that $f(x,f(x,y))=x$ if $f(x,y)=x$ and $f(x,f(x,y))=y$ if $f(x,y)=y$; together,
$f(x,f(x,y))=f(x,y)$. On the other hand, trivially $f(f(x,x),y)=f(x,y)$. Also
note that $f(x,f(y,x))=f(x,f(x,y))=f(x,y)$, so $f$ is a ternary WNU. 
For every $k \geq 3$, $f$ generates a WNU $g_k$ of arity $k$:
$g_k(x_1,\dots,x_k) = f(f(\ldots(f(x_1,x_2),x_3),\ldots),x_k)$.
By Lemma~\ref{lem:suppclone}, $\supp(\Gamma)$ is a clone,
so $g_k \in \supp(\Gamma)$ for all $k \geq 3$.
Therefore, $\supp(\Gamma)$ satisfies the BWC so
the result follows from Theorem~\ref{thm:main}.
\end{proof}

\begin{example}\label{ex:stp-mjn}
In this example we denote by $\multiset{\ldots}$ a multiset.
Let $\omega$ be a binary fractional operation on $D$ defined by
$\omega(f)=\omega(g)=\frac{1}{2}$ and let $\mu$ be a ternary fractional
operation on $D$ defined by $\mu(h_1)=\mu(h_2)=\mu(h_3)=\frac{1}{3}$. Moreover, assume
that $\multiset{f(x,y),g(x,y)}=\multiset{x,y}$ for every $x$ and $y$ and
$\multiset{h_1(x,y,z),h_2(x,y,z),h_3(x,y,z)}=\multiset{x,y,z}$ for every $x$, $y$, and
$z$. 
Let $\Gamma$ be a language on $D$ such that for every two-element subset $\{a,b\} \subseteq D$,
either $\omega|_{\{a,b\}}$ is an STP or $\mu|_{\{a,b\}}$ is an MJN.
Kolmogorov and \v{Z}ivn\'y proved the tractability of $\Gamma$ by a
3-consistency algorithm and a reduction, via Example~\ref{ex:stp}, to submodular
function minimisation~\cite{kz13:jacm}. Such languages were shown to be the only
tractable languages among conservative valued constraint
languages~\cite{kz13:jacm}.
We will discuss conservative valued constraint languages in more detail in
Section~\ref{sub:conservative}.
\end{example}

The following corollary of Theorem~\ref{thm:main} covers Example~\ref{ex:stp-mjn}.

\begin{corollary}\label{cor:stp-mjn}
Let $\Gamma$ be a valued constraint language of finite size with fractional polymorphisms
$\omega$ and $\mu$ as described in Example~\ref{ex:stp-mjn}.
Then, $\Gamma$ has valued relational width $(2,3)$.
\end{corollary}

\begin{proof}
Let $P$ be the set of $2$-element subsets of $D$ such that $\omega|_{\{a,b\}}$
is an STP for $\{a,b\}\in P$ and $\mu|_{\{a,b\}}$ is an MJN for $\{a,b\}\not\in
P$. 
Let $p(x,y,z)=f(f(g(y,x),g(x,z)),g(y,z))$. Observe that $p|_{\{a,b\}}$ is a
majority for $\{a,b\}\in P$, and $p|_{\{a,b\}}$ is either $\proj^{(3)}_1$ or
$\proj^{(3)}_2$ for $\{a,b\}\not\in P$ (possibly different projections for
different $2$-element subsets from $P$). Now let
$q(x,y,z)=p(h_1(x,y,z),h_2(x,y,z),h_3(x,y,z))$. 
For $x,y\in\{a,b\}\in P$, $q(x,x,y)=q(x,y,x)=q(y,x,x)=p(\multiset{x,x,y})=x$.
For $x,y\in\{a,b\}\not\in P$, $q(x,x,y)=q(x,y,x)=q(y,x,x)=p(x,x,y)=x$
as $p$ is either the first or the second projection. Thus, $q$ is
a majority operation. 
The claim then follows from Corollary~\ref{cor:maj}.
\end{proof}

\subsection{Dichotomy for Conservative Valued Constraint Languages}
\label{sub:conservative}

A valued constraint language $\Gamma$ is called \emph{conservative}
if $\Gamma$ contains all unary $\{0,1\}$-valued weighted relations.
Kolmogorov and \v{Z}ivn\'y gave a dichotomy theorem for such languages,
showing that they are either NP-hard, or tractable, cf.~Example~\ref{ex:stp-mjn}.
Here we prove this dichotomy using the SA$(2,3)$-relaxation as the
algorithmic tool.

First, we will need a technical lemma showing that the $\opt$ operator preserves tractability.

\begin{lemma}\label{lem:addopt}
Let $\Gamma$ be a valued constraint language and $I$ an instance of VCSP$(\Gamma)$.
Then, VCSP$(\Gamma \cup \{ \opt(I) \})$ polynomial-time reduces to VCSP$(\Gamma)$.
\end{lemma}

\begin{proof}
Let $\Gamma' = \Gamma \cup \{ \opt(I) \}$.
Let $J' = \sum_{i=1}^q \phi_i(\tup{x}_i)$ be an arbitrary instance of VCSP$(\Gamma')$.
We will create an instance $J$ of VCSP$(\Gamma)$ such that if the optimum of $J$ is
too large, then $J'$ is not satisfiable, and otherwise the optimum of $J'$ can be computed
from the optimum of $J$.
The variables of $J$ are the same as for $J'$.
Let $\phi_i(\tup{x}_i)$ be any valued constraint in $J'$.
If $\phi_i \in \Gamma$‚ then we add the valued constraint $\phi_i(\tup{x}_i)$ to $J$.
Otherwise $\phi_i = \opt(I)$.
In this case, we add $C$ copies of $I(\tup{x}_i)$ to the instance $J$,
where $C$ is a number that will be chosen large enough so that if $J'$ is satisfiable,
then in any optimal assignment to $J$, the variables $\tup{x}_i$ will be forced to
be an optimal solution to the instance $I$.
In such a solution, $\tup{x}_i \in \opt(I)$, and we can recover an optimal solution to $J'$.

The value of $C$ is chosen as follows:
if $I$ does not have any sub-optimal satisfying assignment, then let $C = 1$.
Otherwise, let $C = \lceil(U-L+1)/\Delta\rceil$,
where $U$ is an upper bound on the optimal value of $J'$,
$L$ is a lower bound on the optimal value of $J'$,
and $\Delta$ is the least difference between a sub-optimal and an optimal assignment to $I$.
Both $U$ and $L$ can be computed in polynomial time by taking the sum of the largest,
respectively smallest, finite values of each valued constraint.
The value of $C$ depends linearly on the number of constraints in $J'$, so
the size of $J$ is polynomial in the size of $J'$.

Let $\min(J)$, $\min(J')$, and $\min(I)$ denote the optimal value of the respective instance.
Assume first that $J'$ has a satisfying assignment.
Then, this assignment is also a satisfying assignment to $J$, so
\begin{equation}\label{eq:jprimesat}
\min(J) \leq CN \min(I) + \min(J'),
\end{equation}
where $N$ is the number of occurrences of $\opt(I)$ in $J'$.

If $J$ has a satisfying assignment $\sigma$, then we distinguish two cases.
First, assume that $\sigma$ assigns an optimal value to every copy of $I$.
Then, $\sigma$ is also a satisfying assignment of $J'$, so
\begin{equation}\label{eq:jopt}
\min(J') \leq {\rm Val}(\sigma) - CN\min(I),
\end{equation}
where ${\rm Val}(\sigma)$ denotes the value of $\sigma$.
From (\ref{eq:jprimesat}) and (\ref{eq:jopt}), we see that if $\sigma$ is an optimal
assignment to $J$, so that ${\rm Val}(\sigma) = \min(J)$,
then it is also an optimal assignment to $J'$.

Otherwise, $\sigma$ assigns a sub-optimal value to at least $C$ copies of $I$, so
\begin{equation}\label{eq:jnotopt}
{\rm Val}(\sigma) \geq C(\min(I)+\Delta) + C(N-1) \min(I) +  \min(J') \geq (U-L+1)+CN\min(I)+L.
\end{equation}
In this case, $\min(J) > CN\min(I)+U \geq CN\min(I) + \min(J')$, so by (\ref{eq:jprimesat}),
we see that $J'$ cannot satisfiable.

In summary, if $J$ is unsatisfiable, or if $\min(J) > CN\min(I)+U$, then $J'$ is unsatisfiable,
and otherwise $\min(J') = \min(J) - CN \min(I)$.
\end{proof}

The following theorem was proved by Takhanov~\cite{Takhanov10:stacs} with a 
reduction, essentially amounting to Lemma~\ref{lem:addopt}, added in~\cite{kz13:jacm}.

\begin{theorem}[\hspace*{-0.3em}\cite{kz13:jacm,Takhanov10:stacs}]\label{thm:takhanov-kz}
Let $\Gamma$ be a conservative valued constraint language.
If $\pol(\Gamma)$ does not contain a majority polymorphism, then $\Gamma$ is NP-hard.
\end{theorem}

\begin{theorem}
Let $\Gamma$ be a conservative valued constraint language.
Either $\Gamma$ is NP-hard, or $\Gamma$ has valued relational width $(2,3)$.
\end{theorem}

\begin{proof}
Let $F$ be the set of majority operations in $\pol(\Gamma) \setminus \supp(\Gamma)$.
By Lemma~\ref{lem:killing}, for each $f \in F$, there is an instance $I_f$ of VCSP$(\Gamma)$
such that $f \not\in \pol(\opt(I_f))$.
Let $\Gamma' = \Gamma \cup \{ \opt(I_f) \mid f \in F \}$.
Assume that $\pol(\Gamma')$ contains a majority polymorphism $f$.
Then, $f \not\in F$, so $f \in \supp(\Gamma)$.
From Corollary~\ref{cor:maj}, it follows that $\Gamma$ has
valued relational width $(2,3)$.
If $\pol(\Gamma')$ does not contain a majority polymorphism, then,
since $\Gamma$ is conservative, so is $\Gamma'$, and hence
$\Gamma'$ is NP-hard by Theorem~\ref{thm:takhanov-kz}.
Therefore,
$\Gamma$ is NP-hard by Lemma~\ref{lem:addopt}.
\end{proof}

\section{Conclusions}

We have shown that most previously studied tractable 
valued constraint languages that are not purely relational
fall into the cases covered by Theorem~\ref{thm:main}.
There is however one class of languages which we
have not succeeded in analysing.
These are the valued constraint languages improved by
a so-called generalised weak tournament pair (GWTP)
identified in Uppman~\cite{Uppman13:icalp}.
The definition of this class is rather intricate and we pose
as an open problem the question whether such languages
have valued relational width $(2,3)$.

\begin{problem}
Do valued constraint languages improved by a generalised
weak tournament pair have valued relational width $(2,3)$?
\end{problem}

\bibliographystyle{plain}

\newcommand{\noopsort}[1]{}

\end{document}